\pdfoutput=1
\documentclass[11pt,letterpaper]{llncs}

\usepackage{caption}
\usepackage{amssymb, amsmath} 
\usepackage{mathtools}
\usepackage[margin=1.35in]{geometry}
\usepackage[sort,nocompress]{cite}

\newcommand{\bigO}{\mathcal{O}}
\newcommand{\runs}{r}
\newcommand{\len}{\lambda}
\newcommand{\pos}{\pi}
\newcommand{\NULL}{\bot}

\begin{document}

\title{From LZ77 to the Run-Length Encoded Burrows-Wheeler Transform, and Back}

\author{Alberto Policriti\textsuperscript{1,2} \and Nicola Prezza\textsuperscript{3}\thanks{Part of this work was done while the author was a PhD student at the University of Udine, Italy. Work supported by the Danish Research Council (DFF-4005-00267)}}

\institute{University of Udine, Department of Informatics, Mathematics, and Physics, Italy \and Institute of Applied Genomics, Udine, Italy \and Technical University of Denmark, DTU Compute}

\maketitle

\begin{abstract}
	
	The Lempel-Ziv factorization (LZ77) and the Run-Length encoded Burrows-Wheeler Transform (RLBWT) are two important tools in text compression and indexing, being their sizes $z$ and $\runs$ closely related to the amount of text self-repetitiveness. In this paper we consider the problem of converting the two representations into each other within a working space proportional to the input and the output. Let $n$ be the text length. We 
	show that $RLBWT$ can be converted to $LZ77$ in $\bigO(n\log r)$ time and $\bigO(r)$ words of working space. Conversely, we provide an algorithm to convert $LZ77$ to $RLBWT$ in $\bigO\big(n(\log \runs + \log z)\big)$ time and $\bigO(r+z)$ words of working space. Note that $\runs$ and $z$ can be \emph{constant} if the text is highly repetitive, and our algorithms can operate with (up to) \emph{exponentially} less space than naive solutions based on full decompression.
\end{abstract}

\section{Introduction}\label{section:intro}

The field of \emph{compressed computation}---i.e. computation on compressed representations of the data without first fully decompressing it---is lately receiving much attention due to the ever-growing rate at which data is accumulating in archives such as the web or genomic databases. Being able to operate directly on the compressed data can make an enormous difference, considering that repetitive collections, such as sets of same-species genomes or software repositories, can be compressed at rates that often exceed 1000x. In such cases, this set of techniques makes it  possible to perform most of the computation directly in primary memory and enables the possibility of manipulating huge datasets even on resource-limited machines. 

Central in the field of compressed computation are \emph{compressed data structures} such as compressed full-text indexes, geometry (e.g. 2D range search), trees, graphs. The compression of these structures (in particular those designed for unstructured data) is based on an array of techniques which include entropy compression, Lempel-Ziv parsings~\cite{ziv1977universal,ziv1978compression} (LZ77/LZ78), grammar compression~\cite{charikar2005smallest}, and the Burrows-Wheeler transform~\cite{burrows1994block} (BWT). Grammar compression, Run-Length encoding of the BWT~\cite{siren2009run,siren2012compressed} (RLBWT), and LZ77 have been shown superior in the task of compressing highly-repetitive data and, as a consequence, much research is lately focusing on these three techniques. 

In this paper we address a central point in compressed computation: can we convert between different compressed representations of a text while using an amount of working space proportional to the input/output? Being able to perform such task would, for instance, open the possibility of converting between compressed data structures (e.g. self-indexes) based on different compressors, all within compressed working space. 

It is not the fist time that this problem has been addressed. In~\cite{rytter2003application} the author shows how to convert the LZ77 encoding of a text into a grammar-based encoding, while in~\cite{bannai2012efficient,bannai2013converting} the opposite direction (though pointing to LZ78 instead of LZ77) is considered. In~\cite{tamakoshi2013run} the authors consider the conversions between LZ78 and run-length encoding of the text. Note that LZ77 and run-length encoding of the BWT are much more powerful than LZ78 and run-length encoding of the text, respectively, so methods addressing conversion between LZ77 and RLBWT would be of much higher interest. 
In this work we show how to efficiently solve this problem in space proportional to the sizes of these two compressed representations. See the Definitions section for a formal definition of $RLBWT(T)$ and $LZ77(T)$ as a list of $r$ pairs and $z$ triples, respectively. Let $RLBWT(T)\rightarrow LZ77(T)$ denote the computation of the list $LZ77(T)$ using as input the list $RLBWT(T)$ (analogously for the opposite direction). The following results are illustrated below:

\begin{enumerate}
	\item[(1)] We can compute $RLBWT(T)\rightarrow LZ77(T)$ in $\bigO(n\log \runs)$ time and $\bigO(\runs)$ words of working space
	\item[(2)] We can compute $LZ77(T) \rightarrow RLBWT(T)$ in $\bigO\big(n(\log \runs + \log z)\big)$ time and $\bigO(\runs+z)$ words of working space
\end{enumerate}

Result (1) is based on our own recent work~\cite{policriti2016computing} and requires space proportional to the input \emph{only} as output is streamed to disk. Result (2) requires space proportional to the input \emph{plus} the output, since data structures based on both compressors are used in main memory. In order to achieve result (2), we show how we can (locally) decompress $LZ77(T)$ while incrementally building a run-length BWT data structure of the reversed text. Extracting text from LZ77 is a computationally expensive task, as it requires a time proportional to the parse height $h$ per extracted character~\cite{kreft2013compressing} (with $h$ as large as $\sqrt n$, in the worst case). The key ingredient of our solution is to use the run-length BWT data structure itself to efficiently extract text from $LZ77(T)$. 

\section{Basics}\label{basics}

Since we work with both LZ77~\cite{ziv1977universal} and the Burrows-Wheeler transform~\cite{burrows1994block} (see below for definitions), we assume that our text $T$ contains both \emph{LZ} and \emph{BWT terminator} characters. More specifically, let $T$ be of the form $T=\#T'\$ \in \Sigma^n$, with $T'\in(\Sigma \setminus \{\$,\#\})^{n-2}$,  where $ \$ $ is the LZ77-terminator, and $\#$---lexicographically smaller than all elements in $\Sigma$---is the BWT-terminator. 
Note that adding the two terminator characters to our text increases only by two the number of LZ77 factors and by at most four the number of BWT runs.

The \emph{Burrows-Wheeler Transform}~\cite{burrows1994block} $BWT(T)$ is a permutation of $T$ defined as follows. Sort all cyclic permutations of $T$ in a \emph{conceptual} matrix $M\in\Sigma^{n\times n}$. $BWT(T)$ is the last column of $M$. With $F$ and $L$ we will denote the first and last column of $M$, respectively, and we will say \emph{F-positions} and \emph{L-positions} to refer to positions on these two columns. On compressible texts, $BWT(T)$ exhibits some remarkable properties that permit to boost compression. In particular, it can be shown~\cite{siren2012compressed} that repetitions in $T$ generate equal-letter runs in $BWT(T)$. We can efficiently represent this transform as the list of pairs
$$
RLBWT(T) = \langle  \len_i, c_i \rangle_{i=1,\dots, \runs_T}
$$
where $\len_i>0$ is the length of the \emph{maximal} $i$-th $c_i$-run, $c_i\in\Sigma$. Equivalently, $RLBWT(T)$ is the \emph{shortest} list of pairs $\langle  \len_i, c_i \rangle_{i=1,\dots, \runs_T}$ satisfying $BWT(T) = c_1^{\len_1}c_2^{\len_2}\dots c_{r_T}^{\len_{r_T}}$. Let $\overleftarrow T$ be the reverse of $T$. To simplify notation we define $\runs=\max\{\runs_T, \runs_{\overleftarrow T}\}$ (in practical cases $\runs_T \approx \runs_{\overleftarrow T}$ holds~\cite{belazzougui2015composite}, and this definition simplifies notation).

With $RLBWT^+(T)$ we denote a run-length encoded BWT \emph{data structure} on the text $T$, taking $\bigO(r)$ words of space and supporting \texttt{insert}, \texttt{rank}, \texttt{select}, and \texttt{access} operation on the BWT. Using these operations, functions LF and FL (mapping L-positions to F-positions and \textit{vice versa}) and function \texttt{extend} (turning $RLBWT^+(T)$ into $RLBWT^+(aT)$ for some $a\in\Sigma$) can be supported in $\bigO(\log r)$ time. We leave to the next sections details concerning the particular implementation of this data structure.

We recall that $BWT(\overleftarrow T)$ can be built online with an algorithm that reads $T$-characters left-to-right and inserts them in a dynamic string data structure~\cite{hon2007space,chan2007compressed}.  Briefly, letting $a\in\Sigma$, the algorithm is based on the idea of backward-searching the extended reversed text $\overleftarrow{T\!a}$ in the BWT index for $\overleftarrow T$. This operation leads to the F-position $l$ where $\overleftarrow{T\!a}$ should appear among all sorted $\overleftarrow T$'s suffixes. At this point, it is sufficient to insert $\#$ at position $l$ in $BWT(\overleftarrow T)$ and replace the old $\#$ with $a$ to obtain $BWT(\overleftarrow{T\!a})$.

The \emph{LZ77 parsing}~\cite{ziv1977universal} (or \emph{factorization}) of a text $T$ is the sequence of \emph{phrases} (or \emph{factors}) 
$$LZ77(T) = \langle \pos_i,\len_i,c_i \rangle_{i=1,\dots,z}$$ 
where $ \pos_i\in \{0, \ldots, n-1\}\cup \{\bot\}$ and $ \bot $ stands for ``undefined'', $ \len_i \in \{0, \ldots, n-2\}$, $c_i\in\Sigma$, and:
\begin{enumerate}
	\item $T = \omega_1c_1\ldots \omega_zc_z$, with $\omega_i=\epsilon$ if $\len_i=0$ and $\omega_i=T[\pos_i,\ldots ,\pos_i+\len_i-1]$ otherwise. 
	\item For any $i=1,\ldots ,z$, the string $\omega_i$ is the \emph{longest} occurring at least twice in $\omega_1c_1\ldots \omega_i$.
\end{enumerate}

\section{From RLBWT to LZ77}

Our algorithm to compute $RLBWT(T) \rightarrow LZ77(T)$ is based on the result~\cite{policriti2016computing}: an algorithm to compute---in $\bigO(\runs)$ words of working space and $\bigO(n\log \runs)$ time---$LZ77(T)$ using $T$ as input. The data structure at the core of this result is a dynamic run-length compressed string:

\begin{theorem}\label{th:dynamic RL}~\cite{makinen2010storage,policriti2016computing}
	Let $S\in \Sigma^n$ and let $\bar\runs$ be the number of equal-letter runs in $S$. There exists a data structure taking $\bigO(\bar\runs)$ words of space and supporting \texttt{rank}, \texttt{select}, \texttt{access}, and \texttt{insert} operations on $S$ in $\bigO(\log\bar\runs)$ time.
\end{theorem}

The algorithm works in two steps, during the first of which builds $RLBWT^+(\overleftarrow T)$ by inserting left-to-right $T$-characters in a \emph{dynamic} $RLBWT$ represented with the data structure of Theorem \ref{th:dynamic RL}---using the procedure sketched in the previous section. In the second step, the procedure scans $T$ once more left-to-right while searching (reversed) LZ77 phrases in $RLBWT^+(\overleftarrow T)$. At the same time, a dynamic suffix array sampling is created by storing, for each BWT equal-letter run, the two most external (i.e. leftmost and rightmost in the run) text positions seen up to the current position; the key property proved in~\cite{policriti2016computing} is that this sparse suffix array sampling is sufficient to locate LZ77 phrase boundaries and sources. LZ77 phrases are outputted in text order, therefore they can be directly streamed to output. The total size of the suffix array sampling never exceeds $2\runs$. 
From Theorem \ref{th:dynamic RL}, all operations (\emph{insert}, \emph{LF-mapping}, \emph{access}) are supported in $\bigO(\log \runs)$ time and the structure takes $\bigO(\runs)$ words of space. The claimed space/time bounds of the algorithm easily follow.

Note that, using the algorithm described in~\cite{policriti2016computing}, we can only perform the conversion $RLBWT^+(\overleftarrow T) \rightarrow LZ77(T)$. Our full procedure to achieve conversion $RLBWT(T) \rightarrow LZ77(T)$ consists of the following three steps:
\begin{enumerate}
	\item convert $RLBWT(T)$ to $RLBWT^+(T)$, i.e. we add support for \texttt{rank}/\texttt{select}/\texttt{access} queries on $RLBWT(T)$;
	\item compute $RLBWT^+(\overleftarrow T)$ using $RLBWT^+(T)$;
	\item run the algorithm described in~\cite{policriti2016computing} and compute $LZ77(T)$ using $RLBWT^+(\overleftarrow T)$.
\end{enumerate}

Let $RLBWT(T) = \langle  \len_i, c_i \rangle_{i=1,\dots, \runs}$  (see the previous section). Step 1 can be performed by just inserting characters $c_1^{\len_1}c_2^{\len_2}\dots c_{\runs}^{\len_{\runs}}$ (in this order) in the dynamic run-length encoded string data structure of Theorem \ref{th:dynamic RL}.
Step 2 is performed by  extracting characters $T[0], T[1], \dots, T[n-1]$ from $RLBWT^+(T)$ and inserting them (in this order) in a dynamic $RLBWT$ data structure with the BWT construction algorithm sketched in the  Section (\ref{basics}). Since this algorithm builds the $RLBWT$ of the \emph{reversed} text, the final result is $RLBWT^+(\overleftarrow T)$. 
We can state our first result: 

\begin{theorem}\label{th:rlbwt-lz77}
	Conversion $RLBWT(T) \rightarrow LZ77(T)$ can be performed in $\bigO(n\log r)$ time and $\bigO(\runs)$ words of working space. 
\end{theorem}
\begin{proof}
	We use the dynamic RLBWT structure of Theorem \ref{th:dynamic RL} to implement components $RLBWT^+(T)$ and $RLBWT^+(\overleftarrow T)$. Step 1 requires $n$ \texttt{insert} operations in  $RLBWT^+(T)$, and terminates therefore in $\bigO(n\log \runs)$ time. Since the string we are building contains $r_T$ runs, this step uses $\bigO(\runs)$ words of working space. Step 2 calls $n$ \texttt{extend} and \texttt{FL} queries on dynamic RLBWTs. \texttt{extend} requires a constant number of \texttt{rank} and \texttt{insert} operations~\cite{chan2007compressed}. 
	FL function requires just an \texttt{access} and a \texttt{rank} on the F column and a \texttt{select} on the L column. 
	From Theorem \ref{th:dynamic RL}, all these operations are supported in $\bigO(\log r)$ time, so also step 2 terminates in $\bigO(n\log \runs)$ time. Recall that $r$ is defined to be the maximum between the number of runs in $BWT(T)$ and $BWT(\overleftarrow T)$. Since in this step we are building $RLBWT^+(\overleftarrow T)$ using $RLBWT^+(T)$, the overall space is bounded by $\bigO(r)$ words. Finally, step 3 terminates in $\bigO(n\log \runs)$ time while using  $\bigO(r)$ words of space~\cite{policriti2016computing}. The claimed bounds for our algorithm to compute $RLBWT(T) \rightarrow LZ77(T)$ follow.
\end{proof}

\section{From LZ77 to RLBWT}\label{sec:lz77->rlbwt}

Our strategy to convert $LZ77(T)$ to $RLBWT(T)$ consists of the following steps:

\begin{enumerate}
	\item extract $T[0], T[1], \dots, T[n-1]$ from $LZ77(T)$ and build  $RLBWT^+(\overleftarrow T)$;
	\item convert $RLBWT^+(\overleftarrow T)$ to $RLBWT^+(T)$;
	\item extract equal-letter runs from $RLBWT^+(T)$ and stream $RLBWT(T)$ to the output.
\end{enumerate}

Step 2 is analogous to step 2 discussed in the previous section. Step 3 requires reading characters $RLBWT^+(T)[0]$, ..., $RLBWT^+(T)[n-1]$ (\texttt{access} queries on $RLBWT^+(T)$) and keeping in memory a character storing last run's head and a counter keeping track of last run's length. Whenever we open a new run, we stream last run's head and length to the output.

The problematic step is the first. As mentioned in the introduction, extracting a character from $LZ77(T)$ requires to follow a chain of character copies. In the worst case, the length $h$ of this chain---also called the parse height (see~\cite{kreft2013compressing} for a formal definition)---can be as large as $\sqrt n$. Our observation is that, since we are building $RLBWT^+(\overleftarrow T)$, we can use this component to  efficiently extract text from $LZ77(T)$: while decoding factor $\langle \pos_v, \len_v, c_v \rangle$, we convert $\pos_v$ to a position on the RLBWT and extract $\len_v$ characters from it. The main challenge in efficiently achieving this goal is to convert text positions to RLBWT positions (taking into account that the RLBWT is dynamic and therefore changes in size and content). 

\subsection{Dynamic functions}

Considering that $RLBWT^+(\overleftarrow T)$ is built incrementally, we need a data structure to encode a  function $\mathcal Z :\{\pi_1,...,\pi_z\} \rightarrow \{0,...,n-1\}$ mapping those text positions that are the source of some LZ77 phrase to their corresponding $RLBWT$ positions. Moreover, the data structure must be \emph{dynamic}, that is it must support the following three operations (see below the list for a description of how these operations will be used): 

\begin{itemize}
	\item \texttt{map}: $\mathcal Z(i)$. Compute the image of $i$
	\item \texttt{expand}: $\mathcal Z.expand(j)$. Set $\mathcal Z(i)$ to $\mathcal Z(i)+1$ for every $i$ such that $\mathcal Z(i)\geq j$
	\item \texttt{assign}: $\mathcal Z(i) \leftarrow j$. Call $\mathcal Z.expand(j)$ and set $\mathcal Z(i)$ to $j$
\end{itemize}

To keep the notation simple and light, we use the same symbol $\mathcal Z$ for the function as well as  for the data structure representing it.
We say that  $\mathcal Z(i)$ is \emph{defined} if, for some $j$, we executed an \texttt{assign} operation $\mathcal Z(i) \leftarrow j$ at some previous stage of the computation. For technical reasons that will be clear later, we restrict our attention to the case where we execute \texttt{assign} operations $\mathcal Z(i) \leftarrow j$ for increasing values of $i$, i.e. if $\mathcal Z(i_1) \leftarrow j_1, \dots, \mathcal Z(i_q) \leftarrow j_q$ is the sequence (in temporal order) of the calls to \texttt{assign}  on $\mathcal Z$, then $i_1 < \dots < i_q$.
This case will be sufficient in our case and, in particular, $i_1, \dots, i_q$ will be the sorted non-null phrases sources $\pi_1,\dots, \pi_z$.  Finally, we assume that $\mathcal Z(i)$ is always called  when  $\mathcal Z(i)$ has already been defined---again, this will be the case in our algorithm. 

Intuitively,  $\mathcal Z.expand(j)$ will be used when we insert $T[i]$ at position $j$ in the partial $RLBWT^+(\overleftarrow T)$ and $j$ is not associated with any phrase source (i.e. $i\neq \pi_v$ for all $v=1,\dots,z$). When we insert $T[i]$ at position $j$ in the partial $RLBWT^+(\overleftarrow T)$ and $i = \pi_v$ for some $v=1,\dots,z$ (possibly more than one), $\mathcal Z(i) \leftarrow j$ will be used. 


\medskip

The existence and associated query-costs of the data structure $\mathcal Z$ are proved in the following lemma.

\begin{lemma}
	Letting $z$ be the number of phrases in the LZ77 parsing of $T$, there exists a data structure  taking $\bigO(z)$ words of space and supporting \texttt{map}, \texttt{expand}, and \texttt{assign} operations on $\mathcal Z :\{\pi_1,...,\pi_z\} \rightarrow \{0,...,n-1\}$ in $\bigO(\log z)$ time	 
\end{lemma}
\begin{proof}
	
	First of all notice that, since $LZ77(T)$ is our input, we know beforehand the  domain $\mathcal D = \{ \pi\ |\ \langle\pi, \len, c\rangle \in LZ77(T)\ \wedge \pi\neq \bot \}$ of $\mathcal Z$. We can therefore map the domain to rank space and restrict our attention to functions $\mathcal Z':\{0,...,d-1\} \rightarrow \{0,...,n-1\}$, with $d = |\mathcal D| \leq z$. To compute $\mathcal Z(i)$ we map $0\leq i < n$ to a rank $0\leq i' < d$ by binary-searching a precomputed array containing the sorted values of $\mathcal D$ and return $\mathcal Z'(i')$. Similarly, $\mathcal Z(i) \leftarrow j$ is implemented by executing $\mathcal Z'(i') \leftarrow j$ (with $i'$ defined as above), and $\mathcal Z.expand(j)$ simply as $\mathcal Z'.expand(j)$.
	
	We use a dynamic gap-encoded bitvector $C$ marking (by setting a bit) those positions $j$ such that $j=\mathcal Z(i)$ for some $i$. A dynamic gap-encoded bitvector with $b$ bits set can easily be implemented using a red-black tree such that it takes $\bigO(b)$ words of space and supports \texttt{insert}, \texttt{rank}, \texttt{select}, and \texttt{access} operations in $\bigO(\log b)$ time; see~\cite{policriti2016computing} for such a reduction. 
	Upon initialization of $\mathcal Z$, $C$ is empty. Let $k$ be the number of bits set in $C$ at some step of the computation. 
	We can furthermore restrict our attention to \emph{surjective} functions $\mathcal Z'':\{0,...,d-1\} \rightarrow \{0,...,k-1\}$ as follows. $\mathcal Z'(i')$ (\texttt{map}) returns $C.select_1(\mathcal Z''(i'))$. The \texttt{assign} operation $\mathcal Z'(i') \leftarrow j$ requires the \texttt{insert} operation $C.insert(1,j)$ followed by the execution of $\mathcal Z''(i') \leftarrow C.rank_1(j)$. Operation $\mathcal Z'.expand(j)$ is implemented with $C.insert(0,j)$.
	
	To conclude, since we restrict our attention to the case where---when calling $\mathcal Z(i) \leftarrow j$---argument $i$ is greater than all $i'$ such that $\mathcal Z(i')$ is defined, we will execute \texttt{assign} operations $\mathcal Z''(i') \leftarrow j''$ for increasing values of $i'=0,1,\dots,d-1$. In particular, at each \texttt{assign} $\mathcal Z''(i') \leftarrow j''$, $i'= k$ will be the current domain size. We therefore focus on a new operation, \texttt{append}, denoted as $\mathcal Z''.append(j'')$ and whose effect is $Z''(k) \leftarrow j''$. We are  left with the problem of finding a data structure for a  \emph{dynamic permutation} $\mathcal Z'':\{0,...,k-1\} \rightarrow \{0,...,k-1\}$ with support for \texttt{map} and \texttt{append} operations. Note that both domain and codomain size ($k$) are incremented by one after every \texttt{append} operation.
	
	\begin{example}
		Let $k=5$ and $\mathcal Z''$ be the permutation $\langle 3,1,0,4,2 \rangle$. After  $\mathcal Z''.append(2)$, $k$ increases to $6$ and $\mathcal Z''$ turns into the permutation $\langle 4,1,0,5,3,2\rangle$. Note that $\mathcal Z''.append(j'')$ has the following effect on the permutation: all numbers larger than or equal to $j''$ are incremented by one, and $j''$ is appended at the end of the permutation.
	\end{example}

	To implement the dynamic permutation $\mathcal Z''$, we use a red-black tree $\mathcal T$. We associate to each internal tree node $x$ a counter storing the number of leaves contained in the subtree rooted in $x$. Let $m$ be the size of the tree. The tree supports two operations: 
	
	\begin{itemize}
		\item $\mathcal T.insert(j)$. Insert a new leaf at position $j$, i.e. the new leaf will be the $j$-th leaf to be visited in the in-order traversal of the tree. This operation  can be implemented using subtree-size counters to guide the insertion. After the leaf has been inserted, we need to re-balance the tree (if necessary) and update at most $\bigO(\log m)$ subtree-size counters. The procedure returns (a pointer to) the tree leaf $x$ just inserted. Overall, $\mathcal T.insert(j)$ takes $\bigO(\log m)$ time
		\item $\mathcal T.locate(x)$. Take as input a leaf in the red-black tree and return the (0-based) rank of the leaf among all leaves in the in-order traversal of the tree. $\mathcal T.locate(x)$ requires climbing the tree from $x$ to the root and use subtree-size counters to retrieve the desired value, and therefore runs in $\bigO(\log m)$ time.  
	\end{itemize}

	At this point, the dynamic permutation $\mathcal Z''$ is implemented using the tree described above and a vector $N$ of red-black tree leaves supporting \texttt{append} operations (i.e. insert at the end of the vector). $N$ can be implemented with a simple vector of words with initial capacity 1. Every time we need to add an element beyond the capacity of $N$, we re-allocate $2|N|$ words for the array. $N$ supports therefore constant-time access and amortized constant-time append operations. Starting with empty $\mathcal T$ and $N$, we implement operations on  $\mathcal Z''$ as follows: 
	
	\begin{itemize}
		\item $\mathcal Z''.map(i)$ returns $\mathcal T.locate(N[i])$
		\item $\mathcal Z''.append(j)$ is implemented by calling $N.append(\mathcal T.insert(j))$
	\end{itemize}
	
	
	Taking into account all components used to implement our original dynamic function  $\mathcal Z$, we get the bounds of our lemma.
\end{proof}

\subsubsection{The algorithm}

The steps of our algorithm to compute $RLBWT^+(\overleftarrow T)$ from  $LZ77(T)$ are the following:

\begin{enumerate}
	\item  sort $\mathcal D = \{ \pi\ |\ \langle\pi, \len, c\rangle \in LZ77(T)\ \wedge \pi\neq \bot \}$;
	\item process $\langle \pos_v, \len_v, c_v \rangle_{v=1,...,z}$ from the first to last triple as follows. When processing $\langle \pos_v, \len_v, c_v \rangle$:
	\begin{enumerate}
		\item use our dynamic function $\mathcal Z$ to convert text position $\pos_v$ to RLBWT position $j'=\mathcal Z(\pos_v)$
		\item extract $\len_v$ characters from RLBWT starting from position $j'$ by using the LF function; at the same time, extend RLBWT with the extracted characters. 
		\item when inserting a character at position $j$ of the RLBWT, if $j$ corresponds to some text position $i\in \mathcal D$, then  update $\mathcal Z$ accordingly by setting $\mathcal Z(i)\leftarrow j$. If, instead, $j$ does not correspond to any text position in $\mathcal D$,  execute $\mathcal Z.expand(j)$.
	\end{enumerate}
	
\end{enumerate}

Our algorithm is outlined below as Algorithm 1. Follows a detailed description of the pseudocode and a result stating  its complexity.

\medskip

In Lines 1-5 we initialize all structures and variables. In order: we compute and sort set $\mathcal D$ of phrase sources, we initialize current text position $i$ ($i$ is the position of the character to be read), we initialize an empty RLBWT data structure (we will build $RLBWT^+(\overleftarrow T)$ online), and we create an empty dynamic function data structure $\mathcal Z$. In Line 6 we enter the main loop iterating over LZ77 factors. If the current phrase's source is not empty (i.e. if the phrase copies a previous portion of the text), we need to extract $\len_v$ characters from the RLBWT. First, in Line 8 we retrieve the RLBWT position $j'$ corresponding to text position $\pos_v$ with a \texttt{map} query on $\mathcal Z$. Note that, if $\pos_v\neq\NULL$, then $i>\pos_v$ and therefore $\mathcal Z(\pos_v)$ is defined (see next). We are ready to extract characters from RLBWT. For $\len_v$ times, we repeat the following procedure (Lines 10-19). We read the $l$-th character from the source of the $v$-th phrase (Line 10) and insert it in the RLBWT (Line 11). Importantly, the \texttt{extend} operation at Line  11 returns the RLBWT position $j$ at which the new character is inserted; RLBWT position $j$ correspond to text position $i$. We now have to check if $i$ is the source of some LZ77 phrase. If this is the case (Line 12), then we link text position $i$ to RLBWT position $j$ by calling a \texttt{assign} query on $\mathcal Z$ (Line 13). If, on the other hand, $i$ is not the source of any phrase, then we call a \texttt{expand} query on $\mathcal Z$ on the codomain element $j$. Note that, after the \texttt{extend} query at Line 11, RLBWT positions after the $j$-th are shifted by one. If $j'$ is one of such positions, then we increment it (Line 17). Finally, we increment text position $i$ (Line 19). At this point, we finished copying characters from the $v$-th phrase's source (or we did not do anything if the $v$-th phrase consists of only one character). We therefore extend the RLBWT with the $v$-th trailing character (Line 20), and (as done before) associate text position $i$ to RLBWT position $j$ if $i$ is the source of some phrase (Lines 21-24). We conclude the main loop by incrementing the current position $i$ on the text (Line 25). Once all characters have been extracted from LZ77, RLBWT is a run-length BWT structure on $\overleftarrow T$. At Line 26 we convert it to $RLBWT^+(T)$ (see previous section) and return it as a series of pairs $\langle  \len_v, c_v \rangle_{v=1,\dots, \runs}$.

\begin{figure}[h!]
	\begin{center}
		\includegraphics[trim=3cm 9cm 3cm 3.5cm, clip=true, width=1\textwidth]{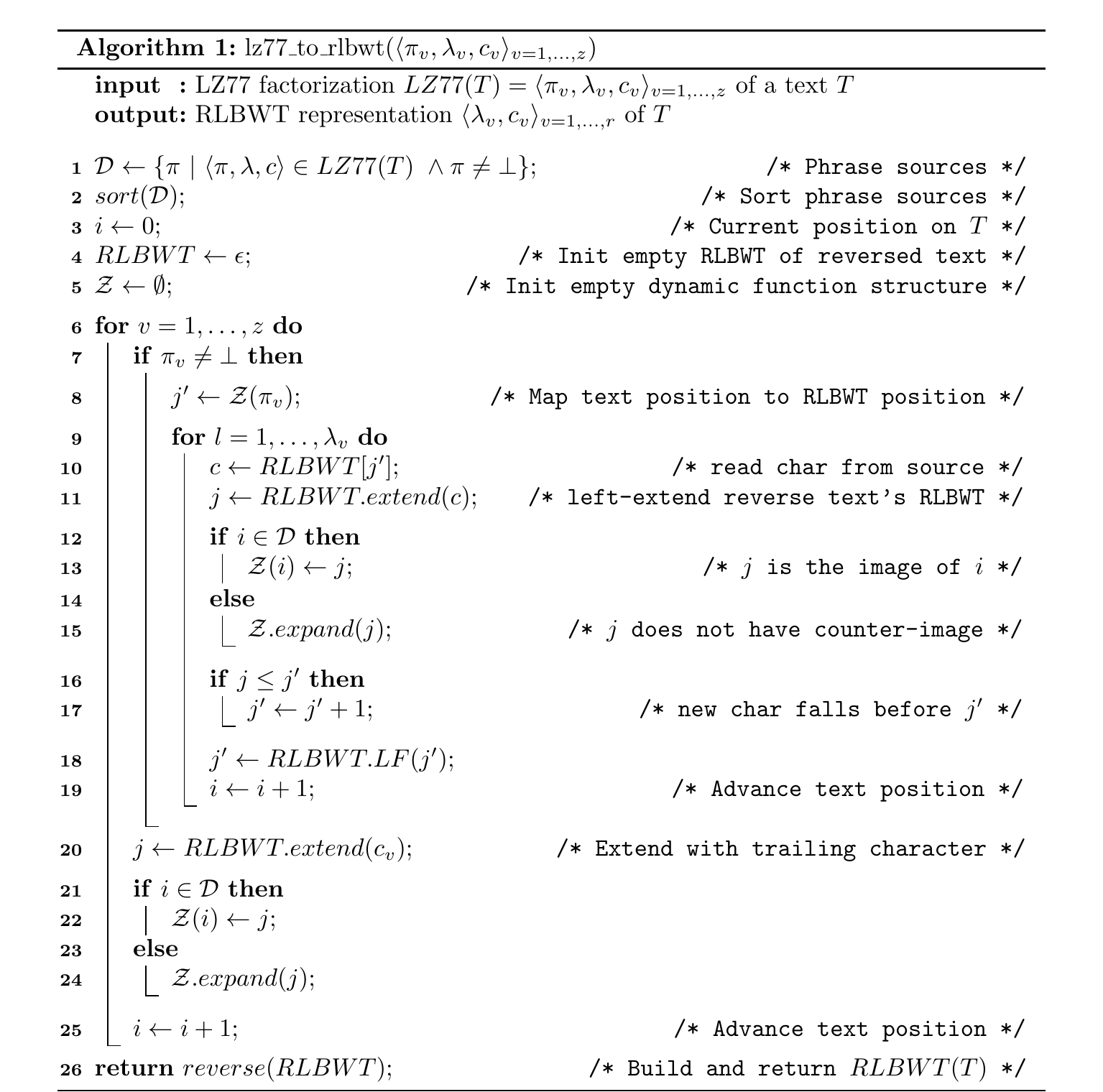}
	\end{center}
\end{figure}

\begin{theorem}\label{th:lz77-rlbwt}
	Algorithm 1 converts $LZ77(T) \rightarrow RLBWT(T)$ in $\bigO(n(\log r+\log z))$ time and $\bigO(\runs+z)$ words of working space
\end{theorem}
\begin{proof}
	Sorting set $\mathcal D$ takes $\bigO(z\log z) \subseteq  \bigO(n\log z)$ time. Overall, we perform $\bigO(z)$ \texttt{map}/\texttt{assign} and $n$ \texttt{expand} queries on $\mathcal Z$. All these operations take globally $\bigO(n\log z)$ time. We use the structure of Theorem \ref{th:dynamic RL} to implement $RLBWT^+(T)$ and $RLBWT^+(\overleftarrow T)$. We perform $n$ \texttt{access}, \texttt{extend}, and \texttt{LF} queries on $RLBWT^+(\overleftarrow T)$. This takes overall $\bigO(n\log \runs)$ time. Finally, inverting $RLBWT^+(\overleftarrow T)$ at Line 26 takes $\bigO(n\log \runs)$ time and $\bigO(r)$ words of space (see previous section). We keep in memory the following structures: $\mathcal D$, $\mathcal Z$, $RLBWT^+(\overleftarrow T)$, and $RLBWT^+(T)$. The bounds of our theorem easily follow.
\end{proof}

\section{Conclusions}

In this paper we presented space-efficient algorithms converting between two compressed file representations---the run-length Burrows-Wheeler transform (RLBWT) and the Lempel-Ziv 77 parsing (LZ77)---using a working space proportional to the input and the output. Both representations can be significantly (up to exponentially) smaller than the text; our solutions are therefore particularly useful in those cases in which the text does not fit in main memory but its compressed representation does. Another application of the results discussed in this paper is the optimal-space construction of compressed self-indexes based on these compression techniques (e.g.~\cite{belazzougui2015composite}) taking as input the RLBWT/LZ77 \emph{compressed} file. 

We point out two possible developments of our ideas. First of all, our algorithms rely heavily on dynamic data structures. On the experimental side, it has been recently shown~\cite{prezza2017framework} that algorithms based on compressed dynamic strings can be hundreds of times slower than others not making use of dynamism (despite offering very similar theoretical guarantees). This is due to factors ranging from cache misses to memory fragmentation; dynamic structures inherently incur into these problems as they need to perform a large number of memory allocations and de-allocations. A possible strategy for overcoming these difficulties is to build the RLBWT by merging two static RLBWTs while using a working space proportional to the output size. A second improvement over our results concerns theoretical running times. We note that our algorithms perform a number of steps proportional to the size $n$ of the text. Considering that the compressed file could be \emph{exponentially} smaller than the text, it is natural to ask whether it is possible to perform the same tasks in a time proportional to $r+z$. This seems to be a much more difficult goal due to the intrinsic differences among the two compressors---one is based on suffix sorting, while the other on replacement of repetitions with pointers.

\bibliographystyle{splncs}
\bibliography{rlbwt-lz77.bib}

\end{document}